\documentclass{llncs}
\usepackage{amsmath}
\usepackage{amssymb}
\usepackage{graphicx}
\usepackage{color}
\usepackage{epstopdf}
\usepackage{subfigure}
\newcommand{\hide}[1]{}
\usepackage{lineno}

\usepackage{soul} 

\usepackage{amsthm}
\usepackage[linesnumbered,ruled]{algorithm2e}
\usepackage{paralist}
\usepackage{tikz}
\usetikzlibrary{calc,shapes,positioning}
\tikzset
{
  VSt/.style =
  {
    circle,  
    inner sep=-0.5,
    minimum size = 1.5mm,
    fill = black!10,
    draw = black,
  }
}

\tikzset
{
  SinkSt/.style =
  {
    rectangle,  
    inner sep=-0.5,
    minimum size = 1.1mm,
    fill = black!80,
    draw = black,
  }
}

\tikzset
{
  PointSt/.style =
  {
    circle,  
    inner sep=-0.5,
    minimum size = 0.8mm,
    fill = black!80,
    draw = black,
  }
}

\tikzset
{
  TickStH/.style =
  {
    rectangle,  
    inner sep=-0.5,
    minimum height = 0.5pt,
    minimum width = 1mm,
    fill = black,
    draw = black,
  }
}

\tikzset
{
  TickStV/.style =
  {
    rectangle,  
    inner sep=-0.5,
    minimum height = 1mm,
    minimum width = 0.5pt,
    fill = black,
    draw = black,
  }
}

\usepackage{float}
\newfloat{floatTogether}{th!}{ext}

\begin{document}
\title{Minmax Regret 1-Sink for Aggregate\\ Evacuation Time on Path Networks\thanks{This work was supported
in part by NSERC Discovery Grants, awarded to B. Bhattacharya,
in part by JST CREST Grant Number JPMJCR1402 held by N. Katoh and Y. Higashikawa,
 and in part by JSPS Kakenhi Grant-in-Aid for Young Scientists (B) (17K12641) given to Y. Higashikawa.
}}%
\author{Binay Bhattacharya$^1$ \and Yuya Higashikawa$^2$ \and \\
Tsunehiko Kameda$^1$ \and Naoki Katoh$^3$
}
\author{Binay Bhattacharya\inst{1} \and Yuya Higashikawa\inst{2} \and\\
Tsunehiko Kameda\inst{1} \and Naoki Katoh\inst{3}
}
\institute{
School of Computing Science, Simon Fraser Univ., Canada
\and 
School of Business Administration, Univ. of Hyogo, Kobe, Japan
\and
School of Science \& Technology, Kwansei Gakuin Univ., Sanda, Japan}

\maketitle

\begin{abstract}
Evacuation in emergency situations can be modeled by a dynamic flow network.
Two criteria have been used before: one is the evacuation completion time and the other is the
aggregate evacuation time of individual evacuees.
The aim of this paper is to optimize the aggregate evacuation time in the simplest case, 
where the network is a path and only one evacuation center
(called a sink) is to be introduced.
The evacuees are initially located at the vertices,
but their precise numbers are unknown, and are given by upper and lower bounds.
Under this assumption,
we compute the sink location that minimizes the maximum ``regret.''
We present an $O(n^2\log n)$ time algorithm to solve this problem,
improving upon the previously fastest $O(n^3)$ time algorithm,
where $n$ is the number of vertices.
\end{abstract}

\section{Introduction}
Investigation of evacuation problems dates back many years~\cite{hamacher2002,mamada2002}.
The goal is to evacuate all the evacuees to some sinks to optimize a certain objective function.
The problem can be modeled by a dynamic flow network whose vertices represent the places where 
the evacuees are initially located and the edges represent possible evacuation routes.
Associated with each edge is the transit time across the edge and its capacity
in terms of the number of people who can traverse it per unit time~\cite{hamacher2002}.
A {\em completion time} {\em $k$-sink}, a.k.a. {\em min-max} {\em $k$-sink},
 is a set of $k$ sinks that minimizes
the time until every evacuee evacuates to a sink.
If the edge capacities are uniform,
it is straightforward to compute a completion time 1-sink in a path network in linear time,
as shown by Cheng and Higashikawa {\em et al.}~\cite{cheng2013,higashikawa2015c}.
Mamada {\em et al.}~\cite{mamada2006} solved this problem for a dynamic tree network
with non-uniform edge capacities in $O(n\log^2 n)$ time.
Higashikawa {\em et al.} proposed an $O(n\log n)$ algorithm for a tree network with
uniform edge capacities~\cite{higashikawa2014b}.

The concept of {\em regret} was introduced by Kouvelis and Yu~\cite{kouvelis1997},
to model the situations where optimization is required when the exact values (such
as the number of evacuees at the vertices) are unknown.
Their model only assumes that the upper and lower bounds on those values are known.
The objective is to find a solution which is as good as any other solution in the worst case,
where the actual values are the most unfavorable.

Motivated by the 2011 earthquake in Japan,
Cheng {\em et al.}~\cite{cheng2013} applied minmax regret optimization to
the completion time 1-sink problem to model evacuation whose objective function is the completion time,
and proposed an $O(n\log^2 n)$ time algorithm for dynamic flow path networks with uniform
edge capacities.
There has been a flurry of research activities on this problem since then.
The initial result was soon improved to $O(n\log n)$,
independently by Higashikawa {\em et al.}~\cite{higashikawa2015c} and 
Wang~\cite{wang2014},
and further to $O(n)$ by Bhattacharya and Kameda \cite{bhattacharya2015b}.
Li {\em et al.}~\cite{li2016b} propose an $O(n^3\log n)$ time algorithm to find
the minmax regret completion time 2-sink problem on dynamic flow path networks.
For the $k$-sink version of the problem,
Arumugam {\em et al.}~\cite{arumugam2014} give two algorithms,
which run in $O(kn^3\log n)$ and $O(kn^2(\log n)^k)$ time,
respectively.
As for dynamic flow tree networks with uniform edge capacities,
Higashikawa {\em et al.}~\cite{higashikawa2014b} propose an $O(n^2\log^2 n)$ time algorithm for finding
the minmax reget 1-sink.
An $O(n^3 \log n)$ time algorithm for dynamic flow cycle networks with uniform edge capacities
is reported by Xu and Li \cite{xu2015a}.

The objective function we adopt in this paper is the {\em aggregate evacuation time,}
i.e., the sum of the evacuation time of every evacuee,
a.k.a. {\em minsum}~\cite{higashikawa2017a}.
It is equivalent to minimizing the mean evacuation time,
and is motivated by the desire to minimize the transportation cost of evacuation
and the total amount of psychological duress suffered by the evacuees, etc.
It is more difficult than the {\em completion time} (a.k.a. {\em minmax}) variety because
the objective cost function is not unimodal.
It is shown by Benkoczi et al.~\cite{benkoczi2018a} that an aggregate time $k$-sink can be found
 in $O(kn\log^3 n)$ if edge capacities are uniform.
Our aim in this paper is to determine an aggregate time sink that minimizes regret~\cite{averbakh1997}.
The main contribution of this paper to to improve the time complexity
from $O(n^3)$ in \cite{averbakh1997} to $O(n^2\log n)$.
We need to consider $O(n^2)$ scenarios,
which are called {\em pseudo-bipartite} scenarios~\cite{higashikawa2017a}.
We make use of two novel ideas.
One is used in Sec.~\ref{sec:sinks} to compute an aggregate time sink under
each of the $O(n^2)$ scenarios in amortized $O(\log n)$ time per sink.
The other is used in Sec.~\ref{sec:regret} to compute the upper envelope of $O(n^2)$
regret functions (with $O(n^3)$ linear segments in total) in $O(n^2\log n)$ time.

In the next section, 
we define the terms that are used throughout this paper.
We also review some known facts which are relevant to later discussions.
Sec.~\ref{sec:clusters} introduces preprocessing which makes later operations
more efficient.
In Sec.~\ref{sec:sinks}
we show how to compute an aggregate time sink under scenarios that matter.
We then  compute in Sec.~\ref{sec:regret} the optimum sink that minimizes the max regret.

\section{Preliminaries}
\subsection{Notations/definitions}\label{sec:defs}
Let $P(V,E)$ denote a given path network,
where we assume that the vertices in its vertex set $V=\{v_1, v_2, \ldots, v_{n}\}$
are arranged from left to right horizontally.
For $i=1,\ldots, n-1$, there is an edge $e_i=(v_i,v_{i+1})\in E$,
whose length is denoted by $d(e_i)$.
We write $p\in P$ for any point $p$ that is either at a vertex or on an edge of $P$. 
For two points $a,b \in P$, we write $a \prec b$ or $b \succ a$ if $a$ lies to the left of $b$.
The distance between them is denoted by $d(a,b)$.
If $a$ and/or $b$ lies on an edge, the distance is prorated. 
The capacity (the upper limit on the flow rate in each edge) of each edge is $c$ (persons per unit time),
and the transit time per unit distance by $\tau$.

In general, $w(v_i)\in \mathbb{Z}_+$ (the set of the positive integers) refers to the weight of vertex $v_i$,
which represents the number of evacuees initially located at $v_i$.
Under {\em scenario $s$}, 
vertex $v_i$ has a weight $w^s(v_i)$ such that $w^s(v_i) \in [\underline{w}(v_i), \overline{w}(v_i)]$,
where $\underline{w}(v_i)$ (resp. $\overline{w}(v_i)$) is the lower (resp. upper) limit on $w(v_i)$,
satisfying $0< \underline{w}(v_i) \le \overline{w}(v_i)$.
We define the Cartesian product 
\[
{\cal S}\triangleq  \prod_{i=1}^{n} [\underline{w}(v_i), \overline{w}(v_i)]. 
\]
Our objective function is the sum of the evacuation times of all the individual evacuees to point $x$.

\medskip\noindent
More definitions:
\begin{eqnarray}
\Phi^s_L(x) &\triangleq& \mbox{the cost at~} x \mbox{~for the evacuees from the vertices on~}
P[v_1,v_i],\nonumber\\
	&&\mbox{where~} v_i\prec x \preceq v_{i+1}  \nonumber\\
\Phi^s_R(x) &\triangleq& \mbox{the cost at~}  x  \mbox{~for the evacuees from the vertices on~} P[v_{i+1},v_n],\nonumber\\
	&&\mbox{where~}v_i\preceq x \prec v_{i+1}  \nonumber\\
\Phi^s(x) &\triangleq& \Phi^s_L(x) + \Phi^s_R(x)  \nonumber\\
\mu^s &\triangleq& \mbox{\rm argmin}_x \Phi^s(x): \mbox{minsum sink under~} s \nonumber
\end{eqnarray}
\begin{eqnarray}
R^s(x) &\triangleq& \Phi^s(x) - \Phi^s(\mu^s) \mbox{{\em :~regret} at~} x \mbox{~under~} s \nonumber\\
\text{(We say}&&\hspace{-3mm}\text{that scenario $s'$ {\em dominates} scenario $s$ at point $x$ if
$R^{s'}(x) \geq R^s(x)$ holds.)} \nonumber\\
R_{max}(x) &\triangleq& \max_{s\in \cal S}  R^s(x) \mbox{:~max regret at $x$}\nonumber\\
\overline{s}_i &\triangleq& \mbox{the {\em bipartite} scenario under which~}
w(v_j) = \overline{w}(v_j)  \mbox{~for all~} j\leq i
\mbox{~and~} \nonumber\\
&& w(v_j) = \underline{w}(v_j) \mbox{~for all~} j>i, \mbox{~where~} 0\le i \le n \nonumber\\
\underline{s}_i &\triangleq& \mbox{the bipartite scenario under which~} w(v_j) = \underline{w}(v_j)  \mbox{~for all~} j\leq i
\mbox{~and~} \nonumber\\
&& w(v_j) = \overline{w}(v_j) \mbox{~for all~} j>i \nonumber\\
s_0 &\triangleq& \overline{s}_0 = \underline{s}_n, s_M \triangleq \overline{s}_n =  \underline{s}_0\nonumber\\
W^s[v_i] &\triangleq& \sum_{k=1}^i w^s(v_k) \nonumber
\end{eqnarray}
Clearly, we can precompute $\underline{W}[v_i] \triangleq  W^{s_0}[v_i]~(=\sum_{k=1}^i \underline{w}(v_k))$
and  $\overline{W}[v_i] \triangleq W^{s_M}[v_i]~(= \sum_{k=1}^i \overline{w}(v_k))$
in $O(n)$ time for all $i$, $1\leq i \leq n$.

Evacuation starts from all the vertices at the same time $t=0$.
Our model assumes that the evacuees at all the vertices start evacuation at the same time
at the rate limited by the capacity ($c$ persions per unit time) of the outgoing edge.
It also assumes that all the evacuees at a non-sink vertex who were initially there
or who arrive there later evacuate in the same direction (either to the left or to the right),
i.e., the evacuee flow is {\em confluent}.
We sometimes use the term ``cost'' to refer to the aggregate evacuation time
of a group of evacuees to a certain destination point.

Our overall approach is as follows.
\begin{enumerate}
\item
Compute $\{\mu^s \mid s\in \overline{\cal S}^*\}$,
where $\overline{\cal S}^*$ is defined in Sec.~\ref{sec:backNforth} and $|\overline{\cal S}^*|=O(n^2)$.
This step takes $O(n^2\log n)$ time.
\item
Compute $R_{\it max}(x) = \max \{R^s(x) \mid s\in \overline{\cal S}^*\}$.
This step takes $O(n^2\log n)$ time.
\item
Find point $x=\mu^*$ that minimizes $R_{\it max}(x)$.
This step takes $O(n^2)$ time.
\end{enumerate}

\subsection{Clusters}
Given a point $x\in P$, which is not the sink,
the evacuee flow at $x$ toward the sink is a function of time,
in general, alternating between no flow and flow at the rate of
$c$ (persons per unit time),
which is the capacity of each edge.
A maximal group of vertices that provide uninterrupted flow without any gap forms a {\em cluster}.
Such a cluster observed on edge $e_{k-1}=(v_{k-1},v_k)$ arriving from right
via $v_k$ is called an {\em ${\cal R}^s$-cluster}  {\em with respect to} (any point on) $e_{k-1}$,
including $v_{k-1}$. 
An {\em ${\cal L}^s$-cluster} {\em with respect to} $e_j=(v_j,v_{j+1})$, including $v_{j+1}$,
is similarly defined for evacuees arriving from left if the sink lies to the right of $v_j$.
If a cluster $C$ contains a vertex $v$, the cluster is said to {\em carry} the evacuees from $v$.
The first vertex of a cluster is called its {\em front vertex}.
\begin{itemize}
\item
${\cal C}^s_{R,k}$: sequence of all ${\cal R}^s$-clusters w.r.t. $e_{k-1}$ ($k=2, \ldots, n$).
\item
$C^s_{R,k}(v_i)\triangleq$ ${\cal R}^s$-cluster  w.r.t. $e_{k-1}$ that contains vertex $v_i ~(i\ge k)$.
\item
${\cal C}^s_{L,k}$: sequence of all ${\cal L}^s$-clusters w.r.t. $e_k$ ($k=1, \ldots, n-1$).
\item
$C^s_{L,k}(v_i)\triangleq$ ${\cal L}^s$-cluster  w.r.t. $e_k$ that contains vertex $v_i ~(i\le k)$.
\end{itemize}

Thus $C^s_{R,k}(v_k)$ is the first cluster of ${\cal C}^s_{R,k}$.
The total weight of the vertices contained in cluster $C$ is denoted by $\lambda(C)$.
If $v_h$ and $v_i$ $(v_h\prec v_i)$ are the front vertices of two adjacent clusters in ${\cal C}^s_{R,k}$,
then we have
\begin{equation}
d(v_h,v_i)\tau > \lambda(C^s_{R,k}(v_h))/c.\label{eqn:eqn1}
\end{equation}
Intuitively,
this means that when the evacuee from $v_i$ arrives at $v_h$,
all evacuees carried by $C^s_{R,k}(v_h)$ have left $v_h$ already.
For $v_{k-1}\preceq x \prec v_k$,
let us analyze the cost of $C^s_{R,k}(v_i)$ to reach $x$ from right,
where $v_i\succ v_k$.
For the $\lambda(C^s_{R,k}(v_i))$ evacuees to move to $x$,
let us divide the time required into two parts.
The first part, called the {\em intra cost}~\cite{benkoczi2018a},
is the weighted waiting time before departure from the front vertex of $C^s_{R,k}(v_i)$,
and can be expressed as
\begin{equation}\label{eqn:intraCost1}
I(C^s_{R,k}(v_i)) \triangleq \{\lambda(C^s_{R,k}(v_i))\}^2/2c.
\end{equation}
Intuitively, (\ref{eqn:intraCost1}) can be interpreted as follows.
As far as the waiting and travel time is concerned,
we may assume that all the $\lambda(C^s_{R,k}(v_i))$ evacuees were at the front vertex of
$C^s_{R,k}(v_i)$ to start with.
Since evacuees leave $v_i$ at the rate of $c$,
the mean wait time for an evacuee is $\lambda(C^s_{R,k}(v_i))/2c$
and the total for all the evacuees carried by $C^s_{R,k}(v_i)$
is $\lambda(C^s_{R,k}(v_i))/2c\times\lambda(C^s_{R,k}(v_i))=\{\lambda(C^s_{R,k}(v_i))\}^2/2c$.
Note that the intra cost does not depend on $x$,
as long as $v_{k-1}\preceq x \prec v_k$.
To be exact, the ceiling function must be applied to (\ref{eqn:intraCost1}),
but we omit it for simplicity,
and {\em adopt} (\ref{eqn:intraCost1}) as our intra cost~\cite{cheng2013}.

The second part, called the {\em extra cost}~\cite{benkoczi2018a},
is the total transit time from the front vertex of $C^s_{R,k}(v_i)$ to $x$ for all the evacuees carried by $C^s_{R,k}(v_i)$,
and can be expressed as
\begin{equation}\label{eqn:extraCost1}
E(C^s_{R,k}(v_i))  \triangleq d(x,v_j)\lambda(C^s_{R,k}(v_i))\tau,
\end{equation}
where $v_j~(\succ x)$ is the front vertex of $C^s_{R,k}(v_i)$.
For the evacuees carried by $C^s_{L,k}(v_i)$ moving to the right,
we similarly define  its intra and extra costs for $k=1, \ldots, n-1$,
where $v_i\preceq v_k\prec x \preceq v_{k+1}$.

For $v_{k-1}\preceq x \prec v_k$,
we now introduce a cost function
\begin{eqnarray}
\Phi^s_{R,k}(x) & \triangleq &
 \sum_{C\in {\cal C}^s_{R,k}} d(x,v_i)\lambda(C)\tau 
+ \sum_{C\in {\cal C}^s_{R,k}} \lambda(C)^2/2c. \label{eqn:right1}
\end{eqnarray}
Similarly, for $x$ ($v_k\prec x \preceq v_{k+1}$),  we define
\begin{eqnarray}
\Phi^s_{L,k}(x) & \triangleq&  \sum_{C\in {\cal C}^s_{L,k}} d(v_i,x)\lambda(C)\tau
 + \sum_{C\in {\cal C}^s_{L,k}} \lambda(C)^2/2c. \label{eqn:left1}
\end{eqnarray}
When $v_k$ is clear from the context, or when there is no need to refer to it,
we may write $\Phi^s_R(x)$ (resp. $\Phi^s_L(x)$) to mean $\Phi^s_{R,k}(x)$  (resp. $\Phi^s_{L,k}(x)$). 
The aggregate of the evacuation times to $x$ of all evacuees is given by
\begin{equation}\label{eqn:Phisx}
\Phi^s(x)=\left\{\begin{array}{lll}
                     &\Phi^s_{L,k}(x) + \Phi^s_{R,k+1}(x)	&\text{~for } v_k\prec x \prec v_{k+1}\\
                      &\Phi^s_{L,k-1}(x) + \Phi^s_{R,k+1}(x) &\text{~for }  x= v_k.
                 \end{array}
                 \right.
\end{equation}
A point $x$ that minimizes $\Phi^s(x)$ is called a  {\em aggregate time sink},
a.k.a. {\em minsum sink}, under $s$.
An aggregate time sink shares the following property of a {\em median}~\cite{kariv1979b}.
\begin{lemma}{\rm \cite{higashikawa2015a}}\label{lem:atVertex}
Under any scenario there is an aggregate time sink at a vertex.
\end{lemma}

\begin{example}\label{ex:ex1}
Consider an example path network in Fig.~\ref{fig:ex1},
where a circle represents a vertex whose weight under some scenario $s$ is shown in it, 
and the length of each edge is shown above it.
The capacity of each edge is assumed to be $c=1$.
\begin{figure}[ht]
\centering
\includegraphics[height=8mm]{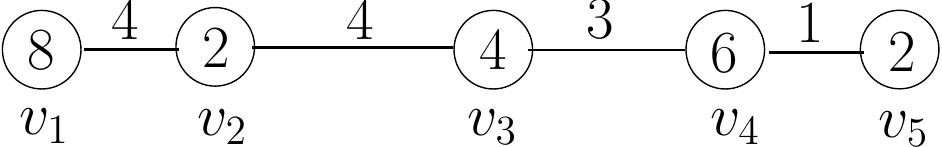}
\caption{An example path network.
}
\label{fig:ex1}
\end{figure}
Let $x$ denote the distance from $v_1$, $d(v_1,x)$,
as well as its position.
Using (\ref{eqn:Phisx}),
we obtain
\begin{eqnarray}\label{eqn:costFunc1}
\Phi^s(x) = 8x + 2(4-x) + 12(8-x) + 32+ 74 = 210 - 6x,
\end{eqnarray}
for $v_1\prec x \prec v_2$,
for example.
Fig.~\ref{fig:Phiofx2} plots $\Phi^s(x)$.
\begin{figure}[ht]
\centering
\includegraphics[height=42mm]{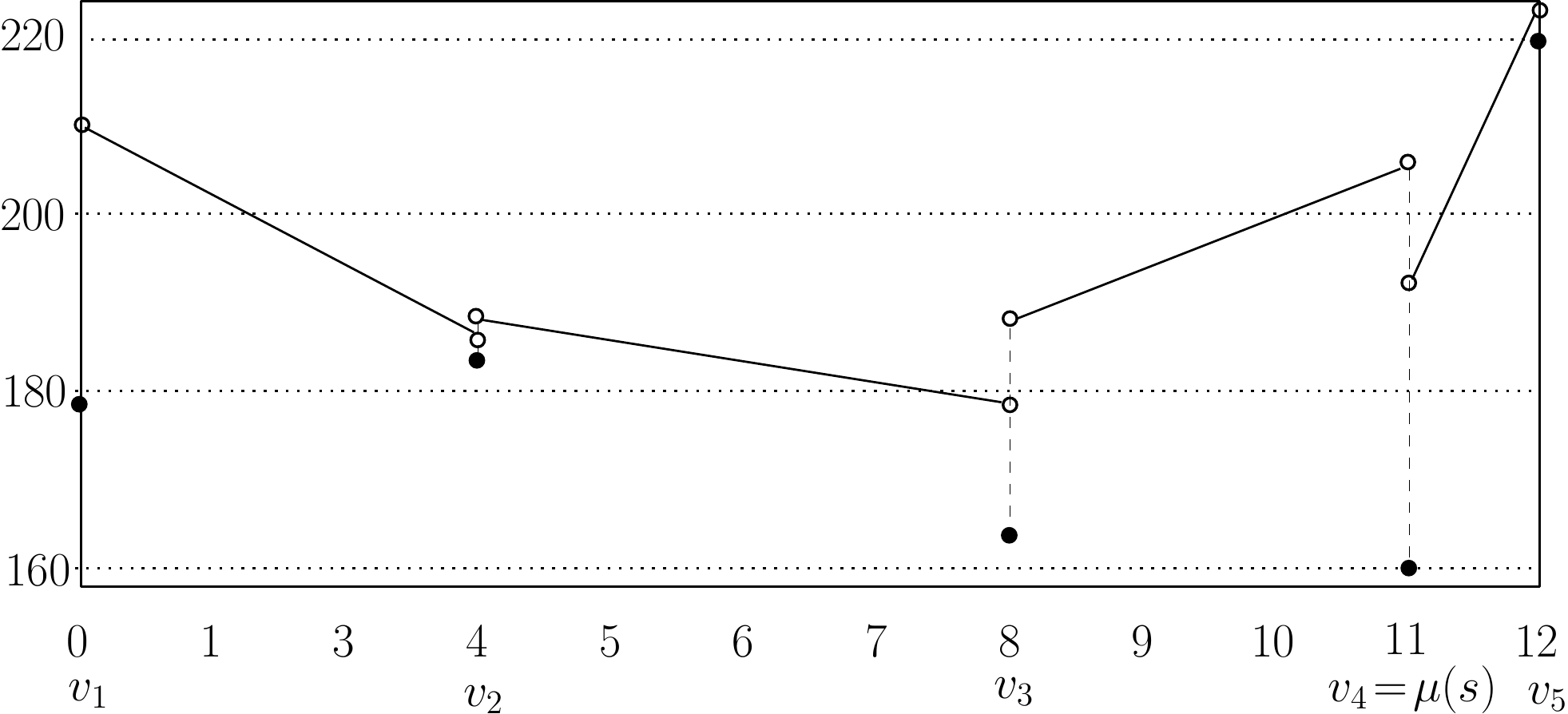}
\caption{Graph for $\Phi^s(x)$.
}
\label{fig:Phiofx2}
\end{figure}
\qed
\end{example}

The above example illustrates the fact that $\Phi^s(x)$ is piecewise linear with discontinuities at the
vertices.
Observe that there is a negative spike at each vertex
because its intra and extra cost contribution is absent,
and that the minsum sink $\mu^s$ is at a vertex,
as stated in Lemma~\ref{lem:atVertex}.

\subsection{What is known}\label{sec:known}
\begin{lemma}{\rm \cite{higashikawa2017a}}\label{lem:minsumsink1}
For any given scenario $s\in {\cal S}$, 
\begin{enumerate}
\item[(a)]
We can compute $\{\Phi^s_L(v_i), \Phi^s_R(v_i)\mid i =1,\ldots,n\}$ in $O(n)$ time.
\item[(b)]
We can compute $\mu^s$ and $\Phi^s(\mu^s)$ in $O(n)$ time.
\end{enumerate}
\end{lemma}

Let $v$ be a vertex and $x$ be a point such that $v \prec x$.
We define a function
\begin{equation}\label{eqn:gamma}
\Gamma^s(x,v) = \Phi^s(x) - \Phi^s(v).
\end{equation}
We have $v=\mu^s$ under some scenario $s$ in mind,
since the regret function can be expressed as $R^s(x) = \Phi^s(x) - \Phi(\mu^s)=\Gamma^s(x,\mu^s)$.
\begin{lemma}{\rm \cite{higashikawa2017a}}\label{lem:bipartite}
For a fixed pair $x,v \in P$,
consider $\Gamma^s(x,v)$ as a funtion of $s$.
Any local maximum of $\Gamma^s(x,v)$ occurs under scenario $s$
under which an adjacent pair of clusters touches each other,
forming a larger cluster.
\end{lemma}

A scenario $s$ under which all vertices on the left (resp. right) of a vertex have
the max (resp. min) weights is called an {\em L-pseudo-bipartite} scenario~\cite{higashikawa2017a}.
The vertex $v_b$, where $1\le b\le n$,
that may take an intermediate weight $w^s(v_b) \in [\underline{w}(v_b), \overline{w}(v_b)]$,
is called the {\em boundary vertex},
a.k.a. {\em intermediate vertex}~\cite{higashikawa2017a}.
Let $b(s)$ denote the index of the boundary vertex under scenario $s$.
We consider the scenarios under which $w(v_b)=\underline{w}(v_b)$ and $w(v_b)=\overline{w}(v_b)$
also as special pseudo-bipartite scenarios,
and in the former (resp. latter) case,
either $b(s)= {b-1}$ or $b(s)=b$ (resp. $b(s)=b$ or $b(s)=b+1$).
The vertices that have the maximum (resp. minmum) weights comprise
the {\em max-weighted part} (resp. {\em min-weighted part}).
We define an {\em R-pseudo-bipartite} scenario symmetrically
with the max-weighted part and the min-weighted part reversed.
As $w(v_b)$ increases from $\underline{w}(v_b)$ to $\overline{w}(v_b)$,
clusters may merge.

Weight $w^s(v_b)$ is said to be a {\em critical weight},
if two clusters with respect to {\em any vertex}
merge as $w(v_b)$ increases to become a scenario $s$.
Let ${\cal S}^*_L$ (resp. ${\cal S}^*_R$) denote the set of the L- (resp. R-)pseudo-bipartite scenarios
that correspond to the critical weights.
Thus each scenario in ${\cal S}^*_L$ (resp. ${\cal S}^*_R$) can be specified by $v_b$ and $w(v_b)$.
Let ${\cal S}^* \triangleq{\cal S}^*_L \cup {\cal S}^*_R$.

\begin{lemma}{\rm \cite{higashikawa2017a}}\label{lem:minsumsink2}
\begin{enumerate}
\item[(a)]
Each scenario in ${\cal S}$ is dominated at every point by a scenario in ${\cal S}^*$.
\item[(b)]
\label{lem:SstarSize} 
$|{\cal S}^*| =O(n^2)$,
and all scenarios in ${\cal S}^*$ can be determined in $O(n^2)$ time.
\end{enumerate}
\end{lemma}
If we use Lemma~\ref{lem:minsumsink1}(b) to find a sink for every scenario in ${\cal S}^*$,
then it takes $O(n^3)$ time.
We will design an algorithm to find them in sub-cubic time in Sec.~\ref{sec:sinks},
after some preparations in Sec.~\ref{sec:clusters}.

\section{Clusters}\label{sec:clusters}
Without loss of generality,
we concentrate on ${\cal R}^s$-clusters,
where $s\in {\cal S}^*_L$.
${\cal L}^s$-clusters and $s\in {\cal S}^*_R$ can be treated analogously.
For $k=2,\ldots, n$,
let ${\cal C}^s_{R,k}$ consist of $q^s(k)$ clusters
\begin{equation} \label{eqn:Rclusters}
{\cal C}^s_{R,k} = \langle C^s_{k,1}, C^s_{k,2}, \ldots, C^s_{k,q^s(k)}\rangle,
\end{equation} 
and let $u_{k,i}$ be the front vertex of $C^s_{k,i}$,\footnote{For simplicity,
we omit subscript $R$ (for $R$ight) and superscript $s$ from $u_{k,i}$.
}
where $v_{k}= u_{k,1}\prec \ldots$ $\prec u_{k,q^s(k)}$.
By (\ref{eqn:eqn1}),
the following holds for $i=1, 2, \ldots, q^s(k)-1$.
\begin{equation}
d(u_{k,i},u_{k,i+1}) > \lambda(C^s_{k,i})/c.\label{eqn:eqn1a}
\end{equation}

\subsection{Preprocessing}\label{sec:precomputation}
\begin{lemma}\label{lem:clusters1}
\begin{enumerate}
\item[(a)]
For any scenario $s\in {\cal S}$,
the number of distinct clusters in $\{{\cal C}^s_{R,k}\mid k=2, \ldots, n\}$ is $O(n)$.
\item[(b)]
For any scenario $s\in {\cal S}$,
we can construct $\{{\cal C}^s_{R,k}\mid k=2, \ldots, n\}$ 
in $O(n)$ time.
\end{enumerate}
\end{lemma}
\begin{proof}
(a) Consider ${\cal C}^s_{R,k}$ in the order $k=n,n-1 \ldots, 2$.
${\cal C}^s_{R,v_n}$ consists of one cluster consisting just of $v_n$.
Let
${\cal C}^s_{R,k+1} = \langle C^s_{k+1,1}, C^s_{k+1,2}, \ldots, C^s_{k+1,q^s(k+1)}\rangle$
for some $k\le n-1$.
The first cluster $C^s_{k,1} \in {\cal C}^s_{R,k}$ contains vertex $v_{k}$
and possibly $C^s_{k+1,1}, \ldots, C^s_{k+1,h}$, 
where $0\le h \le q^s(k+1)$.
$h=0$ means $C^s_{k,1}$ contains just $v_{k}$ and no other vertex.
Note that $C^s_{k,1}$ is new, but the other clusters of ${\cal C}^s_{R,k}$, i.e.,
$C^s_{k,2}, \ldots, C^s_{k,q^s(k)}$ are $C^s_{k+1,h+1}, \ldots, C^s_{k+1,q^s(j)}$,
which are members of ${\cal C}^s_{R,k+1}$.
This means that each $k$ introduces just one new cluster,
and thus the number of distinct clusters is $O(n)$.

(b)
Let us construct ${\cal C}^s_{R,k}$ in the order $k=n,n-1 \ldots, 2$ as in part (a).
Assume that we have computed ${\cal C}^s_{R,k+1}$, and want to compute $C^s_{k,1}$.
If $v_k$ merges with the first $h$ clusters in ${\cal C}^s_{R,k+1}$,
we spend $O(h)$ time in computing $C^s_{k,1}$.
Those $h$ clusters will never contribute to the computing time from now on.
If we pay attention to the front vertex of $C^s_{k,i}$, $u_{k,i}$,
it gets absorbed into a larger cluster at most once,
and each time such an event takes place, constant computation time incurs.
This implies the assertion (b).
\end{proof}
Based on (\ref{eqn:right1}),
we define
\begin{eqnarray}
E^s_{R,k}  &\triangleq& \sum_{C\in {\cal C}^s_{R,k}} d(v_k,v_i)\lambda(C)\tau \label{eqn:extraCost2}\\
I^s_{R,k}  &\triangleq& \sum_{C\in {\cal C}^s_{R,k}} \lambda(C)^2/2c.\label{eqn:intraCost2}
\end{eqnarray}

Computing the extra costs in (\ref{eqn:extraCost2}) is relatively easy,
because it is linear in $\lambda(C)$.
So let us try to compute intra costs efficiently.
As part of preprocessing,
we compute the prefix sum (from left) of the intra costs for the clusters under $s_M$,
and the prefix sum (from right) of the intra costs for the clusters under $s_0$.
To this end we cite the following lemma.
\begin{lemma}{\rm \cite{higashikawa2017a}}\label{lem:EandIcosts}
Given a scenario $s\in {\cal S}$,
\begin{enumerate}
\item[(a)]
We can compute $\{E^s_{R,k}, I^s_{R,k} \mid k=1, \ldots,n-1\}$ in $O(n)$ time.
\item[(b)]
We can compute $\{E^s_{L,k}, I^s_{L,k} \mid k=2, \ldots,n\}$ in $O(n)$ time.
\end{enumerate}
\end{lemma}
The following corollary follows easily from Lemmas~\ref{lem:clusters1} and \ref{lem:EandIcosts}.
\begin{corollary}\label{cor:s0sMclusters}
\begin{enumerate}
\item[(a)]
There are $O(n)$ distinct clusters among
$\{{\cal C}^{s_0}_{L,k}, {\cal C}^{s_0}_{R,k},{\cal C}^{s_M}_{L,k}, {\cal C}^{s_M}_{R,k} \mid k =1,\ldots,n\}$,
and we can compute them in $O(n)$ time. 
\item[(b)]
We can compute $\{E^{s_0}_{R,k}, I^{s_0}_{R,k}, E^{s_M}_{R,k}, I^{s_M}_{R,k} \mid k=1, \ldots,n-1\}$
and $\{E^{s_0}_{L,k}, I^{s_0}_{L,k}, E^{s_M}_{L,k}, I^{s_M}_{L,k} \mid k=1, \ldots,n-1\}$ in $O(n)$ time
\item[(c)]
For each cluster sequence in $\{{\cal C}^{s_0}_{L,k}, {\cal C}^{s_0}_{R,k},{\cal C}^{s_M}_{L,k}, {\cal C}^{s_M}_{R,k}\}$
we can compute the prefix sum of intra costs in $O(n)$ time.
Thus we can compute the prefix sums for all $k$ in $O(n^2)$ time.
\end{enumerate}
\end{corollary}

Let $\overrightarrow{S}^{s_M}[v_j]$ denote the prefix sum from $v_1$ to $v_j$ under $s_M$.
By Corollary~\ref{cor:s0sMclusters}(c),
we can compute them for $j=1,\ldots, n$ in $O(n^2)$ time.
Similarly, $\overleftarrow{S}^{s_0}[v_j]$, the prefix sum from $v_n$ to $v_{j}$ under $s_0$,
can be computed for $j=1,\ldots, n$ in $O(n^2)$ time.
From now on,
we assume that we have computed
all the data mentioned in Corollary~\ref{cor:s0sMclusters},
as well as these prefix sums.

\subsubsection{Constructing ${\cal S}^*$}\label{sec:Sstar}
As observed before each scenario $s\in {\cal S}^*_L$ can be specified by
the boundary vertex $v_b$ and its weight $w(v_b)$.
But a cluster also has another parameter $k$,
as can be seen from (\ref{eqn:Rclusters}).
Let us organize this information by index $k$,
and define\footnote{Note that subscript $L$ of ${\cal S}^*_L$ means that
the left side of $v_b$ is max-weighted,
while the subscript $R$ of $\Delta_{R,k}$ refers to ${\cal R}$-clusters.
}
\begin{equation}\label{eqn:delta2}
\Delta_{R,k} \triangleq \{(b_1,\delta_{k,1}),(b_2,\delta_{k,2}),\ldots\},
\end{equation}
where $b_i\ge k$ for each $i$ and $b_1\le b_2\le \cdots$ hold.
Here $(b_i,\delta_{k,i}) \in \Delta_{R,k}$ means that when $w(v_{b_i})=\delta_{k,i}$
two ${\cal R}^s$-clusters w.r.t. $e_{k-1}$ merge.

Fig.~\ref{fig:deMerge}(a) shows the first R-cluster under $s_M$ with respect to $v_k$,
i.e., $C^{s_M}_{R,k}(v_k)$.
Fig.~\ref{fig:deMerge}(b) shows R-clusters under $s_0$ with respect to $v_k$,
\begin{figure}[ht]
\centering
\includegraphics[height=18mm]{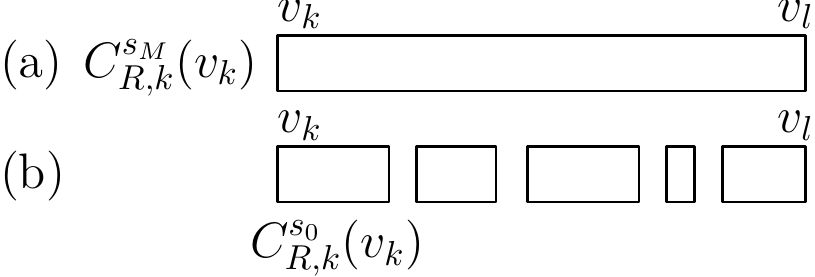}
\caption{(a) $C^{s_M}_{R,k}(v_k)$;
(b)  ${\cal R}^{s_0}$-clusters with respect to $v_k$.
}
\label{fig:deMerge}
\end{figure}
such that the last cluster in Fig.~\ref{fig:deMerge}(b) ends in vertex $v_l$,
which is the last vertex of $C^{s_M}_{R,k}(v_k)$.
Let us start with the clusters in Fig.~\ref{fig:deMerge}(b) and $b=k$.
Suppose we increase $w(v_b)$ by $\delta$ from $\underline{w}(v_b)$
until $C^{s_0}_{R,k}(v_k)$ and the cluster on its right
merge to form a single cluster.
The value of $\delta$ can be obtained by solving
\begin{equation}\label{eqn:delta1}
d(u_{k,1},u_{k,2})= \{\lambda(C^{s_0}_{k,1}) +\delta\}/c.
\end{equation}
If it satisfies
$\underline{w}(v_k) +\delta< \overline{w}(v_k)$,
we can find it in constant time.
Note that for $w^s(v_k)=\lambda(C^{s_0}_{k,1})+\delta$,
$C^{s_0}_{k,1}$ may merge with $C^{s_0}_{k,2}, \ldots, C^{s_0}_{k,h}$,
where $h\ge 2$,
resulting in a combined cluster $C^s_{k,1}$ under $s$,
and the first item $(k, \delta_{k,1})\in \Delta_{R,k}$.
We record $\delta_{k,1}$,
and if $\underline{w}(v_k) +\delta_{k,1}< \overline{w}(v_k)$,
then repeat this operation to find the increment $\delta_{k,2}$, if any, that causes
$C^s_{k,1}$ to merge with $C^{s_0}_{k,h+1}$, etc.
Otherwise, we increment $b$ by one.
When this process terminates,
we will end up with $C^{s_M}_{R,k}(v_k)$,
and we will have constructed $\Delta_{R,k}$ in (\ref{eqn:delta2}).

We formally present the above method to compute\ $\Delta_{R,k}$ as
Algorithm~\ref{alg:alg-1}.
\begin{algorithm}[ht]\label{alg:alg-1}
\KwData {
\begin{compactitem}[--]
\item
${\cal C}^{s_0}_{R,k}=\langle C^{s_0}_{k,1}, C^{s_0}_{k,2}, \ldots, C^{s_0}_{k,q^{s_0}(k)}\rangle$
\item
$v_l=$ the last vertex of $C^{s_M}_{R,k}(v_k)$ \tcp{Precomputed and available}
\item
$\{u_{k,1},u_{k,2}, \ldots, u_{k,q^s(k)}\}$
\end{compactitem}
}
\KwResult {
\begin{compactitem}[--]
\item 
$\Delta_{R,k}$
\end{compactitem}
}
\BlankLine
Set $\Delta_{R,k}=\emptyset$, $C = C^{s_0}_{k,1}$, $b=k$, $h=2$, and $j=1$ ;
 \tcp{\!\!Initialize}
 \While {$b \le l$}{
	\Repeat {$\underline{w}(v_b) + \delta \ge \overline{w}(v_b)$}{ 
	Solve $d(u_{k,1},u_{k,h})=\{\lambda(C) + \delta\}/c$ for $\delta$ \;
		\While{$d(u_{k,1},u_{k,h+1})\le \{\lambda(C\cup C^{s_0}_{k,h}) + \delta\}/c$}{
		$C= C\cup C^{s_0}_{k,h}$  \;
		$h=h+1$
		} 
		Set $\delta_{k,j} = \delta$ and add $(b, {\delta}_{k,j})$ to $\Delta_{R,k}$ \;
		$C= C\cup C^{s_0}_{k,h}$  \;
	 	$j=j+1$ \;				
	  } 
		$b=b+1$ \;
} 
\caption{{\sc Computing} $\Delta_{R,k}$}
\end{algorithm}
Clearly, each item $(b_j,\delta_{k,j}) \in \Delta_{R,k}$ in (\ref{eqn:delta2})
corresponds to a scenario $s_j\in {\cal S}^*_L$ in the following way.
\begin{equation}\label{eqn:Phisx}
w^{s_j}(v_i)= \left\{\begin{array}{lll}
                 &w^{s_M}(v_i)			&\text{~for~} 1\le i< k\\
                 &\underline{w}(v_k)+ \delta_{k,j} 		&\text{~for~} i= k\\
		&w^{s_0}(v_i) 			&\text{~for~} k<i\le n
                 \end{array}
                 \right.
\end{equation}
Let ${\cal S}^*_{L,k}$ be the set of scenarios corresponding to the increments in $\Delta_{R,k}$
according to (\ref{eqn:Phisx}).
Note that under any $s\in {\cal S}^*_{L,k}$,
we have $C^s_{R,k}(v_{b(s)})= C^s_{R,k}(v_k)$.

\begin{lemma}\label{lem:clusters2}
\begin{enumerate}
\item[(a)]
${\cal S}^*_L= \cup^n_{k=1}{\cal S}^*_{L,k}$.
\item[(b)]
Algorithm~\ref{alg:alg-1} runs in $O(|C^{s_M}_{R,k}(v_k)|)$ time,
where $|C|$ denotes the number of vertices in cluster $C$.
\item[(c)]
We can construct $\{\Delta_{R,k}\mid k=2,\ldots, n\}$
in $O(n^2)$ time.
\end{enumerate}
\end{lemma}
\begin{proof}
(a) This is obvious.

(b) We can carry out each step inside the repeat loop of Algorithm~\ref{alg:alg-1} 
in constant time.
We thus spend constant time per cluster of ${\cal C}^{s_0}_{R,k}$ that is contained in $C^{s_M}_{R,k}(v_k)$.

(c) Follows immediately from part (b),
since $O(|C^{s_M}_{R,k}(v_k)|)=O(n)$.
\end{proof}

\subsection{Computing $\Phi^s(v_i)$ for $s\in {\cal S}^*$}\label{sec:Phivi}
Let us now turn our attention to the computation of the extra and intra costs at vertices at the time
when a merger occurs,
namely under the scenarios in ${\cal S}^*_{L,k}$.
While computing $\Delta_{R,k}$ as in Sec.~\ref{sec:precomputation},
we can update the extra and intra costs at $v_k$ under the corresponding scenario in $s\in \tilde{\cal S}^*_{L,k}$
as follows.

When the first increment $\delta_{k,1}$ causes the merger of the first two clusters
$C^{s_0}_{k,1}$ and $C^{s_0}_{k,2}$, for example,
we subtract the extra cost contributions of $C^{s_0}_{k,1}$ and $C^{s_0}_{k,2}$ from $E^{s_0}_{R,k}$,
and add the new contribution from the merged cluster
in order to compute $E^s_{R,k}$ for the new scenario $s$ that results from the incremented weight 
$\underline{w}^s(v_k) =\underline{w}(v_k) +\delta_{k,1}$.
We can similarly compute $I^s_{R,k}$ from $I^{s_0}_{R,k}$ in constant time.
Carrying out these operations whenever a new merged cluster is created thus takes $O(n)$ time
for a given $k$ and $O(n^2)$ time in total for all $k$'s.

Recall the definition of $\overrightarrow{S}^{s_M}[v_j]$ and $\overleftarrow{S}^{s_0}[v_j]$
after Corollary~\ref{cor:s0sMclusters}.
\begin{lemma}\label{lem:costAtv_i}
Assume that all the data mentioned in Corollary~\ref{cor:s0sMclusters} are available.
Then under any given scenario $s\in {\cal S}^*_L$,
we can compute the following in constant time.
\begin{enumerate}
\item[(a)]
$\Phi^s(v_i)= \Phi^s_L(v_i)+\Phi^s_R(v_i)$ for any given index $i$.
\item[(b)]
$\Phi^s(x)= \Phi^s_L(x)+\Phi^s_R(x)$ for any given point $x$.
\end{enumerate}
\end{lemma}
\begin{proof}
(a) 
Let us compute $\Phi^s(v_i)$,
where $v_k\prec v_i \prec v_{b(s)}$.
We already have $\Phi^s_L(v_i)=\Phi^{s_M}_L(v_i)$ available,
so we need $\Phi^s_R(v_i)$.
The difference between ${\cal C}^s_{R,k}$ and ${\cal C}^s_{R,i}$ is illustrated in Fig.~\ref{fig:computePhi}.
Note that the cluster $C^s_{R,k}(v_i)$ may start before $v_i$,
while $C^s_{R,i}(v_i)$ starts at $v_i$.
See the green frames in Fig.~\ref{fig:computePhi}.
Thus more than one cluster of ${\cal C}^s_{R,i}$ may belong to the same cluster in ${\cal C}^s_{R,k}$,
as shown in the figure.
We first determine the last vertex of $C^s_{R,k}(v_i)$ and let it be $v_l$.

Note that the prefix sums were computed only for the critical weights of $v_b$,
and the critical weights for $v_b$ are not the same for ${\cal C}^s_{R,k}$ and ${\cal C}^s_{R,i}$. 
When we use the prefix sum at $v_i$ for the clusters in ${\cal C}^s_{R,i}$,
we should change the weight $w(v_b)$ in Fig.~\ref{fig:computePhi}(b) to that in Fig.~\ref{fig:computePhi}(a). 
This can be done by replacing the intra cost at $v_a$ in Fig.~\ref{fig:computePhi}(b) by that
in Fig.~\ref{fig:computePhi}(a).

There is another possibility that is not covered by Fig.~\ref{fig:computePhi},
namely $C^s_{R,k}(v_i)=C^s_{R,k}(v_b)$,
but $C^s_{R,i}(v_i) \not=C^s_{R,i}(v_b)$ or $C^s_{R,i}(v_i)=C^s_{R,i}(v_b)$.
In this case, there is no vertex $v_a$ in Fig.~\ref{fig:computePhi}(a).
Search for $w^s(v_b)$ among the critical weights for $w(v_b)$ with respect to $v_i$
in Fig.~\ref{fig:computePhi}(b),
and let $w_1 \le w^s(v_b) <w_2$. 
Let $s'$ be the scenario such that  $w^{s'}(v_b)=w_1$, 
and determine the cluster $C^{s'}_{R,i}(v_b)$.
Let $s''$ that results from $s'$ by increasing $w(v_b)$ from $w^{s'}(v_b)~(=w_1)$
to $w^s(v_b)$.
This increase does not affect $C^{s'}_{R,i}(v_b)$,
except that $\Phi^{s''}_R(v_i) >\Phi^{s'}_R(v_i)$ due to the increases in the extra and intra costs.
It is straightforward to compute the increased extra cost.
It is easy to see that $I(C^{s''}_{R,i}(v_b))=\lambda(C^s_{R,k}(v_i))\}^2/2c$.
\begin{figure}[htb]
\centering
\subfigure[${\cal C}^s_{R,k}$]{\includegraphics[height=8mm]{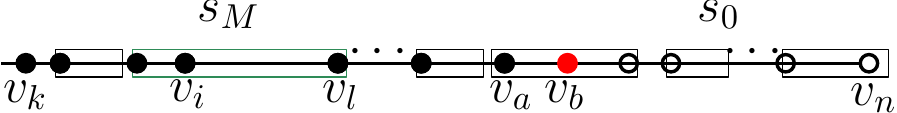}}
\subfigure[${\cal C}^s_{R,i}$]{\includegraphics[height=8mm]{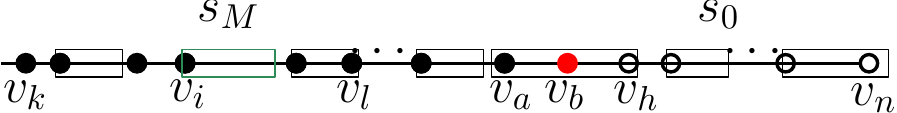}}
\caption{Illustration for the proof of Lemma~\ref{lem:costAtv_i}.
}
\label{fig:computePhi}
\end{figure}
Assume now that Fig.~\ref{fig:computePhi}(b) shows the clusters of ${\cal C}^{s''}_{R,i}$,
which are the same as those of ${\cal C}^{s'}_{R,i}$,
where $v_a=v_i$ is possible.
We can compute $I(C^{s''}_{R,i}(v_b))$ as follows.
\[
I(C^{s''}_{R,i}) = I(C^{s'}_{R,i}) + \{\lambda(C^{s''}_{R,i}(v_b))^2 - \lambda(C^{s'}_{R,i}(v_b))^2\}/2c,
\]
where
\[
\lambda(C^{s''}_{R,i}(v_b))= \lambda(C^{s'}_{R,i}(v_b))+ w^s(v_b) -w_1.
\]
Note that all this takes constant time under the assumption of the lemma.

(b)
Note that we have for $v_{i-1}\preceq x \prec v_i$,
\begin{eqnarray}
\Phi^s_R(x) &=& d(x,v_i)(W^s[v_n] - W^s[v_{i-1}] )\tau + E^s_{R,i} + I^s_{R,i}\nonumber\\ 
		&=& d(x,v_i)(W^s[v_n] - W^s[v_{i-1}] )\tau + \Phi^s_R(v_i), \label{eqn:right2}
\end{eqnarray}
and we can compute $\Phi^s_R(v_i)$ in constant time by part (a).
It is clear that the first term can be computed in constant time.
We can similarly compute $\Phi^s_L(x)$ in constant time.
\end{proof}

\section{Computing sinks $\{\mu^s\mid s\in {\cal S}^*\}$}\label{sec:sinks}
Among the increments in $\Delta_R \triangleq \{\Delta_{R,k}\mid k=2,\ldots, n\}$,
there is a natural lexicographical order,
ordered first by $b$ and then by $w(v_b)$,
from the smallest to the largest.
We write $s\lessdot s'$ if $s$ is ordered before $s'$ in this order.
In what follows we assume the items in $\Delta_R$ are sorted by $\lessdot$.

\subsection{Tracking $\mu^s$}\label{sec:backNforth}
Observe that we have $\Phi^{s}_L(x)=\Phi^{s_M}_L(x)$ for $x\preceq v_b$,
which is independent of $w(v_b)$.
Similarly, we have $\Phi^{s}_R(x)=\Phi^{s_0}_R(x)$ for $x\succeq v_b$,
which is also independent of $w(v_b)$.
We precomputed the piecewise linear function $\Phi^{s_M}_L(x)$ for $x\preceq v_b$,
and $\Phi^{s_0}_R(x)$ for $x\succeq v_b$,
which are independent of $w(v_b)$.
We initialize the current scenario by $s=s_0$,
the boundary vertex $v_b$ by $b=1$,
its weight $w(v_b)= w^{s_0}(v_1)$.
For each successive increment in $\Delta_R$, 
from the smallest (according to $\lessdot$),
we want to know the leftmost (aggregate time) sink under the corresponding scenario.

It is possible that,
as we increase the weight $w(v_b)$,
the sink may jump across $v_b$ from its right side to its left side,
and vice versa, back and forth many times.
We shall see how this can happen below.

By Lemma~\ref{lem:costAtv_i},
for a given index $b$,
we can compute $\{\Phi^{\overline{s}_{b-1}}(v_i) \mid i=1,2,\ldots, n\}$ in $O(n)$ time.\footnote{Recall
the definition of $\overline{s}_j$ from Sec.~\ref{sec:defs}.
}
We first scan those costs at $v_i$ for $i=b, b-1,\ldots, 1$,
and whenever we encounter a vertex with cost smaller than those we examined so far,
we record the index of the vertex.
Let ${\cal I}^b_L$ be the recorded index set.
We then scan those costs at $v_i$ for $i=b+1,\ldots, n$,
and whenever we encounter a vertex with cost smaller than those we examined so far,
we the index of the vertex,
and let ${\cal I}^b_R$ be the recorded index set.
We now plot $p_i=(v_i,\Phi^{\overline{s}_{b-1}}(v_i))$ for $i\in {\cal I}^b_L\cup {\cal I}^b_R$
in the $x$-$y$ coordinate system,
with $v_i$ as the $x$ value and $\Phi^{\overline{s}_{b-1}}(v_i)$ as the $y$ value.
See Fig.~\ref{fig:sinkMovement}.
\begin{figure}[ht]
\centering
\includegraphics[height=32mm]{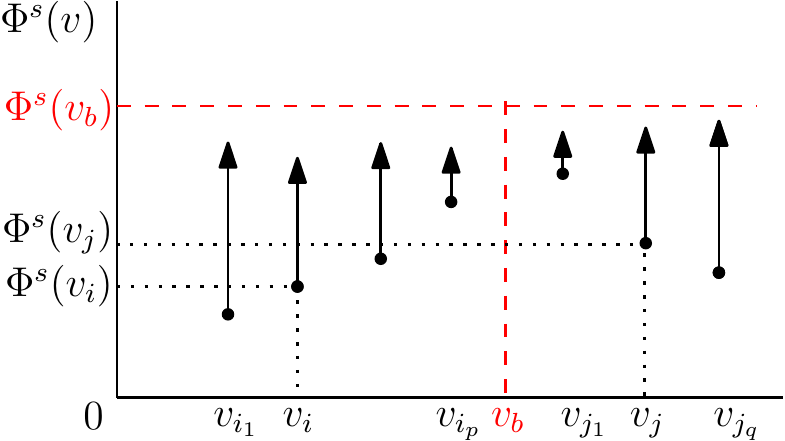}
\caption{2-dimensional representation of
$\Phi^{\overline{s}_{b-1}}(v_i)=\Phi_L^{\overline{s}_{b-1}}(v_i)+\Phi_R^{\overline{s}_{b-1}}(v_i)$.
}
\label{fig:sinkMovement}
\end{figure}
It is clear that for $i,j \in {\cal I}^b_L$,
we have $\Phi^{\overline{s}_{b-1}}(v_i) < \Phi^{\overline{s}_{b-1}}(v_j)$ if $i<j$,
and for $i,j \in {\cal I}^b_R$,
we have $\Phi^{\overline{s}_{b-1}}(v_i) > \Phi^{\overline{s}_{b-1}}(v_j)$ if $i<j$.
Therefore, the points plotted on the left (resp. right) side of $v_b$ get higher and higher
as we approach $v_b$ from left (resp. right),
as can be seen in Fig.~\ref{fig:sinkMovement}.

Note that for a vertex $v_i\prec v_b$,
as $w(v_b)$ is increased,
$\Phi^{s}_R(v_i)$ increases,
while $\Phi^{s}_L(v_i)$ remains fixed at $\Phi^{s_M}_L(v_i)$.
For $v_i\succ v_b$, on the other hand,
as $w(v_b)$ is increased,
$\Phi^{s}_L(v_i)$ increases,
while $\Phi^{s}_R(v_i)$ remains fixed at $\Phi^{s_0}_R(v_i)$.
A vertical arrow in Fig.~\ref{fig:sinkMovement} indicates the amount
of increase in the cost of the corresponding vertex when $w(v_b)$ is increased
by a certain amount.
Note that the farther away a vertex is from $v_b$,
the more is the increase in the cost.

The following proposition follows from the above observations.
\begin{proposition}\label{prop:invariants}
\begin{enumerate}
\item[(a)]
$\Phi^s(v_i) \le \Phi^s(v_j)$ holds for any pair $i,j\in {\cal I}^b_L$ such that $i<j$.
\item[(b)]
$\Phi^s(v_i) \ge \Phi^s(v_j)$ holds for any pair $i,j\in {\cal I}^b_R$ such that $i<j$.
\item[(c)]
Either the vertex with the smallest index in ${\cal I}^b_L$
or the vertex with the largest index in ${\cal I}^b_R$ is a sink, i.e., it has the lowest cost.
\end{enumerate}
\end{proposition}
Note that the cost at $v_b$ is not affected by the change in $w(v_b)$ and remains the same.
We consider the three properties in Proposition~\ref{prop:invariants} as {\em invariant} properties,
and remove the vertices that do not satisfy (a) or (b).
As we increase $w(v_b)$,
in the order of the sorted increments in $\Delta_R$,
we update ${\cal I}^b_L$ and ${\cal I}^b_R$,
looking for the change of the sink.
By property (c), 
the sink cannot move away from $v_b$.
We now make an obvious observation.

\begin{proposition}\label{prop:sameSink}
As $w(v_b)$ is increased,
there is a sink at the same vertex for all increments tested since the last time the sink moved,
until the smallest index in ${\cal I}^b_L$ or
the largest index in ${\cal I}^b_R$ changes,
causing the sink to move.
\end{proposition}

We are thus interested in how ${\cal I}^b_L$ (resp. ${\cal I}^b_R$) change,
in particular when its smallest (resp. largest) index changes.
To find out, let $\delta$ be the smallest increase such that $(b,\delta) \in \Delta_R$
{\em and} increasing $w(v_b)$ by $\delta$ above $\underline{w}(v_b)$
causes the cost of vertex $v_i$ to reach the cost of $v_j$,
where $i$ and $j$ are either adjacent in ${\cal I}^b_L$ and $i<j$ holds,
or adjacent in ${\cal I}^b_R$ and $i>j$ holds.
If such a $\delta$ does not exist,
we set $\delta=\infty$.
Since we can find such a $\delta$ by binary search over $\Delta_R$,
finding it for each adjacent pair of indices in ${\cal I}^b_L$ and ${\cal I}^b_R$
takes $O(\log n)$ time,
and the total time for all adjacent pairs is $O(n\log n)$.
We insert $(\delta;i,j)$ into a min-heap ${\cal H}$,
organized according to the first component $\delta$,
from which we can extract the item with the smallest first component in constant time.
Note that $v_b$ is fixed.

Once ${\cal H}$ has been constructed as above,
we pick the item $(\delta;i,j)$ with the smallest $\delta$ from ${\cal H}$ (in constant time).
If $i,j\in {\cal I}^b_L$ (resp. $i,j\in {\cal I}^b_R$) then we remove $i$ (resp. $j$) from
${\cal I}^b_L$ (resp. ${\cal I}^b_R$),
and compute $(\delta';i^-,j)$ (resp. $(\delta';i,j^+)$) where $i^-$ (resp. $j^+$) is the
index in ${\cal I}^b_L$ (resp. ${\cal I}^b_R$) that is immediately before (resp. after) $i$ (resp. $j$).
We perform binary search to find $\delta'$, taking $O(\log n)$ time,
and insert $(\delta';i^-,j)$ (resp. $(\delta';i,j^+)$) into ${\cal H}$, again taking $O(\log n)$ time.
If $i$ was the smallest index in ${\cal I}^b_L$,
the sink may have moved.
In this case no new item is inserted into ${\cal H}$.
Similarly, if $j$ was the largest index in ${\cal I}^b_R$,
the sink may have moved,
and no new item is inserted into ${\cal H}$.

We repeat this until either ${\cal H}$ becomes empty or the min value in ${\cal H}$ is $\infty$.
It is repeated $O(n)$ times, and the total time required is $O(n\log n)$.
If the sink moves when the smallest index in ${\cal I}^b_L$ or the largest index in ${\cal I}^b_R$
changes,
we have determined the sink under all the scenarios with the lighter $w(v_b)$
since the last time the sink moved. 
Once $w(v_b)=\underline{w}(v_b)+\delta$ reaches $\overline{w}(v_b)$,
$b$ is incremented,
and the new boundary vertex $v_{b+1}$ now lies to the left of the old boundary vertex $v_b$ 
in Fig.~\ref{fig:sinkMovement}.

\subsection{Algorithm}
Algorithm~\ref{alg:alg-2}
is a formal description of our method to find a sink for each increment
of $w(v_b)$,
which are listed in $\Delta_R$.
It refers to ${\cal S}^*_b\triangleq \{s\in {\cal S}^* \mid b(s)=b\}$.
\renewcommand\footnoterule{} 
\begin{floatTogether}
\begin{algorithm}[H]\label{alg:alg-2}
\KwData {
\begin{compactitem}[--]
\item
Boundary vertex $v_b$
\item
Sorted array $\Delta_R$; 
\item 
Index sets ${\cal I}^b_L$ and ${\cal I}^b_R$ for $w(v_b)=\underline{w}(v_b)$
\item
Arrays $\{\underline{W}[\cdot], \overline{W}[\cdot]\}$ and $\{\overrightarrow{S}^{s_M}[v_j], \overleftarrow{S}^{s_0}[v_j]\mid j=1, \ldots, n\}$
\end{compactitem}
}
\KwResult {
\begin{compactitem}[--]
\item 
Sinks $\{\mu^s \mid s\in {\cal S}^*_b\cap {\cal S}^*_L\}$
\end{compactitem}
}
\BlankLine
\For {each adjacent pair $i,j\in{\cal I}^b_L~(i<j)$}{
Using binary search, find the smallest increment $\delta$ in $\Delta_R$ such that
 $\Phi^{s(\delta)}(v_i) \ge \Phi^{s(\delta)}(v_j)$, and insert $(i,j;\delta)$ into a min-heap ${\cal H}_L$\;
 }  
 \For {each adjacent pair $i,j\in{\cal I}^b_R~(i<j)$}{
 Using binary search, find the smallest increment $\delta$ in $\Delta_R$ such that
 $\Phi^{s(\delta)}(v_i) \le \Phi^{s(\delta)}(v_j)$, and insert $(i,j; \delta)$ into a min-heap ${\cal H}_R$\;
 }  
\While{${\cal H}_L \cup{\cal H}_R\not=\emptyset$}{
   From ${\cal H}_L \cup{\cal H}_R$ remove item $(i,j;\delta)$ with smallest $\delta$,
   and name it $\hat{\delta}_{i,j}$ \;
  \If {$i$ is not the first index in ${\cal I}^b_L$ and $j$ is not the last index in ${\cal I}^b_R$}{
 	\If {$(i,j;\delta)\in {\cal H}_L$}
	 {Remove $i^-$ (the immediate predecessor of $i$) from ${\cal I}^b_L$, 
  	 compute $\hat{\delta}_{i^-,j}$, and insert $(i^-,j; \hat{\delta}_{i^-,j})$ into ${\cal H}_L$
    	} 
	\Else {
    	Remove $j^+$ (the immediate successor of $j'$) from ${\cal I}^b_R$,
   	compute ${\hat{\delta}}_{i,j^+}$, and insert $(i,j^+; \hat{\delta}_{i,j^+})$ into ${\cal H}_R$
       	} 
	Skip the {\bf else} part
  } 
\Else {  	
 	\If {$i$ is the first index in ${\cal I}^b_L$}
	 {Remove $i$ from ${\cal I}^b_L$
    	} 
	\If {$j$ is the last index in ${\cal I}^b_R$}
    	{Remove $j$ from ${\cal I}^b_R$
       	} 
  	 If the index corresponding to the current sink has been removed
	 then determine the new sink either at the vertex indexed by first index in ${\cal I}^b_L$ or
	 or the last index in ${\cal I}^b_R$, whichever has the smaller cost \;
	 From now on the sink remains at this new position until it moves the next time
  } 
} 
\caption{{\sc Computing} $\{\mu^s \mid s\in {\cal S}^*_b\cap {\cal S}^*_L\}$}
\end{algorithm}
\end{floatTogether}

\begin{lemma}\label{lem:delta}
\begin{enumerate}
\item[(a)]
The minimum increment in $\Delta_R$ that causes the cost of $v_i$ to exceed that of 
the next vertex closer to $v_b$,
can be determined in $O(\log n)$ time.
\item[(b)]
Algorithm~\ref{alg:alg-2} runs in $O(n\log n)$ time for a given $v_b$.
\end{enumerate}
\end{lemma}
\begin{proof}
(a) Use binary search on $\Delta_R$,
and compare the costs for each probe in constant time.

(b)  Note that $|{\cal S}^*_b|=O(n)$. 
Evaluating $\Phi^{s(\delta)}(v_i)$ and $\Phi^{s(\delta)}(v_j)$ in Lines~2 and 5 takes constant time
by Lemma~\ref{lem:costAtv_i}.
Thus the two for-loops take $O(n\log n)$ time

Updating ${\cal H}_L$ and ${\cal H}_R$ takes $O(\log n)$ time per insertion/deletion,
which will occur at most $n$ times and costs $O(n\log n)$ time.
All other steps take constant time.
Step~8 takes $O(n)$ time.
\end{proof}

For the ${\cal R}^s$-clusters w.r.t. $e_{i-1}$ that lie to the right of $C^s_{R,i}(v_b)$
and are not merged as a result of increase in $w(v_b)$,
the sum of their intra costs was already precomputed.
We can similarly compute $\{\mu^s \mid s\in {\cal S}^*_b\cap {\cal S}^*_R\}$ in $O(n\log n)$ time.
Running Algorithm~\ref{alg:alg-2} and its counterpart for ${\cal S}^*_R$ for $b=1,2,\ldots,n$,
we get
\begin{lemma}\label{lem:allSinks}
The sinks $\{\mu^s \mid s\in {\cal S}^*\}$ can be computed in $O(n^2\log n)$ time.
\end{lemma}

\section{Minmax regret sink}\label{sec:regret}
Since we know the sinks $\{\mu^s\mid s\in{\cal S}^*\}$
(Lemma~\ref{lem:allSinks}),
we proceed to compute the upper envelope for the $O(n^2)$ regret functions
$\{R^s(x) = \Phi^s(x) - \Phi^s(\mu^s)\mid s \in {\cal S}^*\}$.
The minmax regret sink $\mu^*$ is at the lowest point of this upper envelope.

\subsection{Upper envelope for $\{R^s(x) \mid s\in {\cal S}^*\}$}
If we try to find the upper envelope $\max_{s\in {\cal S}^*}\Phi^s(x)$ in one shot,
it would take at least $O(n^3)$ time,
since $|{\cal S}^*|= O(n^2)$,
and for each $s$, $\Phi^s(x)$ consists of $O(n)$ linear segments.
Recall the definition ${\cal S}^*_b= \{s\in {\cal S}^*\mid b(s)=b\}$
for $b=1,2,\ldots, n$.
We employ the following two-phase processing,
and carried out each phase in $O(n^2\log n)$ time.
\begin{itemize}
\item
[{\em Phase 1}:]
 For each $b$, compute the upper envelope $\max_{s\in{\cal S}^*_b}R^s(x)$.
\item
[{\em Phase 2}:] Compute the upper envelope for the results from {\em Phase 1}.
\end{itemize}
In {\em Phase 1}, we successively merge regret functions,
spending amortized $O(\log n)$ time per regret function.
Thus the total time for a given $b$ is $O(n\log n)$ and the total time for all 
$k$ is $O(n^2\log n)$.
In {\em Phase 2}, we then compute the upper envelope for the resulting $O(n)$ regret functions
with a total of $O(n^2)$ linear segments.
To implement {\em Phase 1},
we first prove the following lemma in the Appendix.
\begin{lemma}\label{lem:increasing}
Let $s, s' \in {\cal S}^*_b$ be two scenarios such that and $s\lessdot s'$.
As $x$ moves to the right,
the difference $D(x)=\Phi^{s'}(x) -\Phi^{s}(x)$ decreases monotonically for $v_1 \preceq x\preceq v_b$
and increases monotonically for $v_b \preceq x \preceq v_n$.
\end{lemma}

We divide each regret function in $\{R^s(x) \mid s\in {\cal S}^*_b\}$ into two parts:
left of $v_b$ and right of $v_b$.
We then find the upper envelope for the left set and right set separately.
Note that each $R^s(x)$ has $O(n)$ bending points, since they bend only at vertices.
Taking the max of two such functions may add one extra bending point on an edge,
so the total bending points in the upper bound is still  $O(n)$.

By definition we have
\begin{eqnarray}\label{eqn:diff}
R^{s'}(x) - R^s(x) &=& \Phi^{s'}(x) -\Phi^{s'}(\mu^{s'}) - \{\Phi^s(x)  -\Phi^s(\mu^s)\}\nonumber\\
				&=& \Phi^{s'}(x) -\Phi^s(x)  - \{\Phi^{s'}(\mu^{s'}) -\Phi^s(\mu^s)\}.
\end{eqnarray}
Note that the second term in (\ref{eqn:diff}) is independent of position $x$.
Lemma~\ref{lem:increasing} implies
\begin{lemma}\label{lem:crossing2}
Let $s, s' \in {\cal S}^*_b$ be two scenarios such that and $s\lessdot s'$.
Then $R^{s'}(x)$ may cross $R^{s}(x)$ at most once in the interval $[v_1,v_b]$ from above,
and at most once in the interval $[v_b,v_n]$ from below.
\end{lemma}

See Fig.~\ref{fig:intersectPath2} for an illustration or Lemma~\ref{lem:crossing2}.
\begin{figure}[ht]
\centering
\includegraphics[height=22mm]{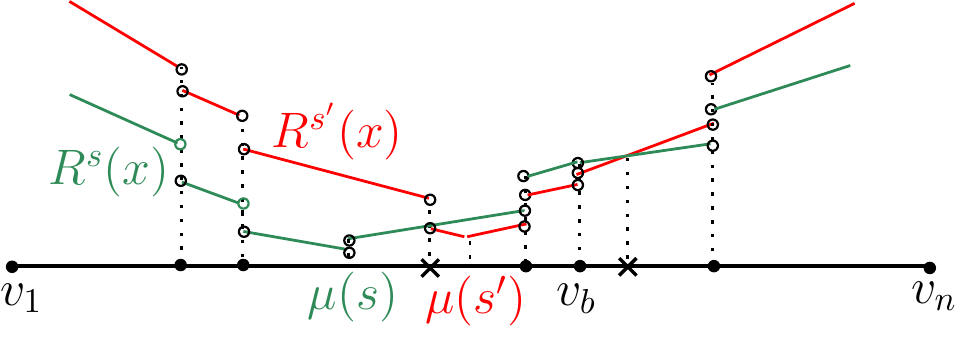}
\caption{$R^s(x)$ and $R^{s'}(x)$ cross each other at {\tt x}'s.}
\label{fig:intersectPath2}
\end{figure}
Algorithm~\ref{alg:alg-3} computes $\max_{s\in {\cal S}^*_b}R^s(x)$.
\begin{lemma}\label{lem:upperEnv}
\begin{enumerate}
\item[(a)]
The upper envelope $\max_{s\in {\cal S}^*_b}R^s(x)$
has $O(|{\cal S}^*_b|+ n)$ line segments.
\item[(b)]
 Algorithm~\ref{alg:alg-3} computes it correctly in $O(|{\cal S}^*_b|\log n)$ time.
\end{enumerate}
\end{lemma}
\begin{proof}
(a) Without loss of generality,
let us consider the upper envelope in the interval $[v_b,v_n]$.
Since $R^s(x) =\Phi^s(x)  -\Phi^s(\mu^s)$,
$R^s(x)$ is linear over the edge connecting any adjacent pair of vertices,
and $\max_{s\in {\cal S}^*_b}\Phi^s(x)$ has $O(|{\cal S}^*_b|+n)$ line segments on
all edges by Lemma~\ref{lem:crossing2}.

(b)  By Lemma~\ref{lem:increasing},
the condition of Line 5 can be tested by their values at $v_b$,
and the condition of Line 8 can be tested by their values at $v_n$.
If $R^{s}(x)$ and $R^{s'}(x)$ in Lemma~\ref{lem:crossing2} intersect at point $X$ to the right of $v_b$,
then we have $R^{s'}(x)\ge R^{s}(x)$ holds for $x\succ X$,
and we can ignore $R^{s}(x)$ for $x\succ X$.
After the {\bf if}$-${\bf else},
Algorithm~\ref{alg:alg-3} ignores the regret function that was processed,
which is also justified by Lemma~\ref{lem:increasing}. 
\end{proof}

\begin{algorithm}[ht]\label{alg:alg-3}
\KwData {
\begin{compactitem}[--]
\item
$\{\Phi^s(x), \mu^s \mid s\in {\cal S}^*_b\}$
\end{compactitem}
}
\KwResult {
\begin{compactitem}[--]
\item 
$\max_{s\in {\cal S}^*_b} R^s(x)$
\end{compactitem}
}
\BlankLine
Order $\{R^s(x)\mid s\in {\cal S}^*_b\}$ by $\lessdot$ from the ``smallest''
to the ``largest'' \;
Initialize the upper envelope $U(x)$ to the ``smallest'' function,
and the {\em crossing point} $X$ to $v_n$ \; 
\While{there is an unprocessed regret function}{
	Pick the next ``smallest'' unprocessed regret function, $R^s(x)$ \;
	\If {$\forall x: R^s(x)\ge U(x)$ (compare their values at $v_b$)}{ 
		set $U(x)=R^s(x)$ and $X=v_n$  and go to Line 17
	 } 
	\If {$\forall x: R^s(x)\le U(x)$ (compare their values at $v_n$)}{ 
		Do nothing and go to Line 17
	 } 
	\If {$U(X)\ge R^s(X)$}{ 
  	   Compute the intersection of $U(x)$ and $R^s(X)$ (to the right of $X$) by binary search,
	   and update $U(x)$ and $X$
	 } 
	\Else {Compute  the intersection of $U(x)$ and $R^s(X)$ (to the left of $X$) by binary search,
	and update $U(x)$ and $X$
	} 
	Mark $R^s(x)$ as {\em ``processed.''}
} 
\caption{{\sc Computing} $\max_{s\in {\cal S}^*_b} R^s(x)$}
\end{algorithm}

\subsection{Main theorem}
Since $O(\sum_{b=1}^n |{\cal S}^*_b|\log n)= O(n^2\log n)$,
Lemma~\ref{lem:upperEnv} implies
\begin{lemma}\label{lem:bendingPoints} 
The upper envelope $\max_{s\in {\cal S}^*}R^s(x)$ has $O(n^2)$ linear segments,
and can be computed in  $O(n^2\log n)$ time.
\end{lemma}
Hershberger~\cite{hershberger1989} showed that the upper envelope of
$m$ line segments can be computed in $O(m\log m)$ time.
We can use his mothod to compute the global upper envelope in $O(n^2\log n)$ time.
So far we didn't pay any attention to the spikes at vertices.
Divide the problem in two subproblems: optimal sink is on an edge, and at a vertex.
Compare the two solutions and pick the better one.
In addition to Lemma~\ref{lem:bendingPoints}, we should evaluate the maximum cost at each vertex.
The minmax regret sink is at the point with the minimum of these maximum costs.
Corollary~\ref{cor:s0sMclusters} and
Lemmas~\ref{lem:minsumsink2}, \ref{lem:allSinks} and \ref{lem:bendingPoints} imply our main result.
\begin{theorem}
The minmax regret sink on a dynamic path network can be computed in $O(n^2\log n)$ time.
\end{theorem}

\section{Conclusion}
We presented an $O(n^2\log n)$ time algorithm for finding the minmax regret aggregate time sink
on dynamic path networks with uniform edge capacities,
which improves upon the previously most efficient $O(n^3)$ time algorithm in
\cite{higashikawa2017a}.
This was achieved by two novel methods.
One was used to compute 1-sinks under the $O(n^2)$ pseudo-bipartite scenarios in amortized $O(\log n)$ time
per scenario, 
and the other was used to compute the upper envelope of $O(n^2)$ regret functions in $O(n^2\log n)$ time.
Note that $O(n^2)$ regret functions have $O(n^3)$ linear segments.
Future research topics include solving the minmax regret problem for aggregate time sink for more general
networks such as trees.

\begin{thebibliography}{10}
\providecommand{\url}[1]{\texttt{#1}}
\providecommand{\urlprefix}{URL }

\bibitem{arumugam2014}
Arumugam, G.P., Augustine, J., Golin, M., Srikanthan, P.: A polynomial time
  algorithm for minimax-regret evacuation on a dynamic path. arXiv:1404,5448v1
  [cs.DS] 22 Apr 2014  165 (2014)

\bibitem{averbakh1997}
Averbakh, I., Berman, O.: Minimax regret $p$-center location on a network with
  demand uncertainty. Location Science  5,  247--254 (1997)

\bibitem{benkoczi2018a}
Benkoczi, R., Bhattacharya, B., Higashikawa, Y., Kameda, T., Katoh, N.: Minsum
  $k$-sink on dynamic flow path network. In: Proc. IWOCA2018, To appear (2018)

\bibitem{bhattacharya2015b}
Bhattacharya, B., Kameda, T.: Improved algorithms for computing minmax regret
  sinks on path and tree networks. Theoretical Computer Science  607,  411--425
  (Nov 2015)

\bibitem{cheng2013}
Cheng, S.W., Higashikawa, Y., Katoh, N., Ni, G., Su, B., Xu, Y.: Minimax regret
  1-sink location problem in dynamic path networks. In: Proc. Annual Conf. on
  Theory and Applications of Models of Computation (T-H.H. Chan, L.C. Lau, and
  L. Trevisan, Eds.), Springer-Verlag, LNCS 7876. pp. 121--132 (2013)

\bibitem{hamacher2002}
Hamacher, H., Tjandra, S.: Mathematical modelling of evacuation problems: a
  state of the art. in: Pedestrian and Evacuation Dynamics, Springer Verlag,
  pp. 227--266 (2002)

\bibitem{hershberger1989}
Hershberger, J.: Finding the upper envelope of $n$ line segments in ${O}(n\log
  n)$ time. Information Processing Letters  33(4),  169--174 (1989)

\bibitem{higashikawa2015c}
Higashikawa, Y., Augustine, J., Cheng, S.W., Golin, M.J., Katoh, N., Ni, G.,
  Su, B., Xu, Y.: Minimax regret 1-sink location problem in dynamic path
  networks. Theoretical Computer Science  588(11),  24--36 (2015)

\bibitem{higashikawa2017a}
Higashikawa, Y., Cheng, S.W., Kameda, T., Katoh, N., Saburi, S.: Minimax regret
  1-median problem in dynamic path networks. Theory of Computing Systems  (May
  2017), {D}OI: 10.1007/s00224-017-9783-8

\bibitem{higashikawa2014b}
Higashikawa, Y., Golin, M.J., Katoh, N.: Minimax regret sink location problem
  in dynamic tree networks with uniform capacity. Journal of Graph Algorithms
  and Applications (18.4),  539--555 (2014)

\bibitem{higashikawa2015a}
Higashikawa, Y., Golin, M.J., Katoh, N.: Multiple sink location problems in
  dynamic path networks. Theoretical Computer Science  607(1),  2--15 (2015)

\bibitem{kariv1979b}
Kariv, O., Hakimi, S.: An algorithmic approach to network location problems,
  {Part II}: The $p$-median. SIAM J. Appl. Math.  37,  539--560 (1979)

\bibitem{kouvelis1997}
Kouvelis, P., Yu, G.: Robust Discrete Optimization and its Applications. Kluwer
  Academic Publishers, London (1997)

\bibitem{li2016b}
Li, H., Xu, Y., Ni, G.: Minimax regret 2-sink location problem in dynamic path
  networks. J. of Combinatorial Optimization  31,  79--94 (2016)

\bibitem{mamada2002}
Mamada, S., Makino, K., Fujishige, S.: Optimal sink location problem for
  dynamic flows in a tree network. IEICE Trans. Fundamentals  E85-A,
  1020--1025 (2002)

\bibitem{mamada2006}
Mamada, S., Uno, T., Makino, K., Fujishige, S.: An ${O}(n\log^2 n)$ algorithm
  for a sink location problem in dynamic tree networks. Discrete Applied
  Mathematics  154,  2387--2401 (2006)

\bibitem{wang2014}
Wang, H.: Minmax regret 1-facility location on uncertain path networks.
  European J. of Operational Research  239(3),  636--643 (2014)

\bibitem{xu2015a}
Xu, Y., Li, H.: Minimax regret 1-sink location problem in dynamic cycle
  networks. Information Processing Letts.  115(2),  163--169 (2015)

\end{thebibliography}

\section*{Appendix}

\subsection{Proof of Lemma~\ref{lem:increasing}}
\noindent
{\bf Lemma 14}
{\em Let $s, s' \in {\cal S}^*_b$ be two scenarios such that and $s\lessdot s'$.
As $x$ moves to the right,
the difference $D(x)=\Phi^{s'}(x) -\Phi^{s}(x)$ increases monotonically for $v_b \preceq x \preceq v_n$,
and decreases monotonically for $v_1 \preceq x\preceq v_b$.}

\begin{proof}
Without loss of generality, we assume that $v_b \preceq x \preceq v_n$,
since essentially the same proof works if $v_1 \preceq x\preceq v_b$.
Let us first consider the extra cost.
If the sum of the vertex weights on the left side of $x$ is larger than that on the right side,
then the extra cost component of $\Phi^{s'}(x)$ grows faster than that of $\Phi^{s}_R(x)$,
Otherwise,
it decreases more slowly.

We now consider the intra costs.
They do not change as long as $x$ moves on the same edge,
including the vertex at its right end.
So we assume that $x$ moves across a vertex, $v_k$,
as illustrated in Fig.~\ref{fig:clusters3a},
where $s$ and $s'$ are two scenarios such that $s \lessdot s'$ and
the both have the same boundary vertex $v_b$.
Let $v_b\in C^s_{L,j}(v_j)$ hence $v_b\in C^{s'}_{L,j}(v_j)$.
Let $v_i$ be the front vertex of the ${\cal L}^s_x$-cluster immediately to the left
of $C^s_{L,j}(v_j)$.
\begin{figure}[ht]
\centering
\includegraphics[width=7.2cm]{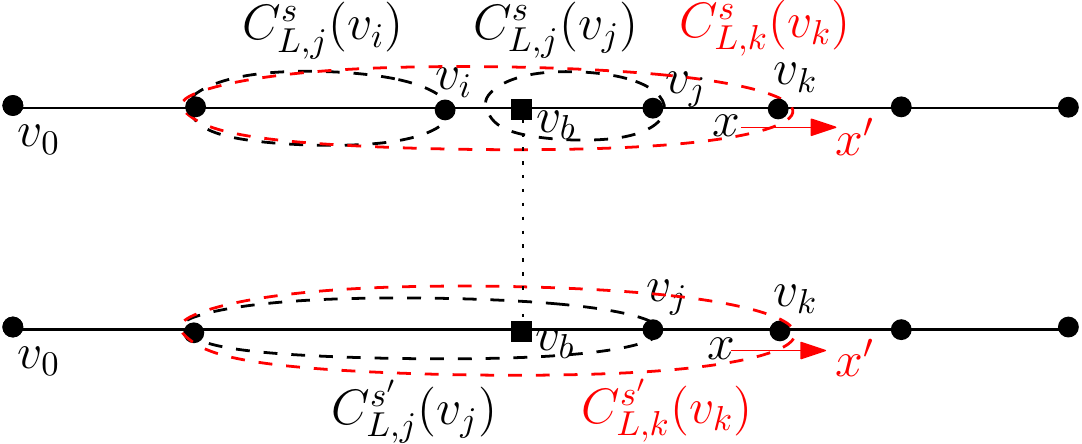}
\caption{${\cal L}^s$-clusters and ${\cal L}^{s'}$-clusters w.r.t. $x$ are shown in dashed black ovals,
and ${\cal L}^s$-clusters and ${\cal L}^{s'}$-clusters w.r.t. $x'$ are shown in dashed red ovals.
}
\label{fig:clusters3a}
\end{figure}
We compare the increase in $\Phi^{s'}(x)$ with that in $\Phi^s(x)$,
as $x$ moves past $v_k$ to $x'$,
and show that the increase in $\Phi^s(x)$ is smaller than that in $\Phi^{s'}(x)$.
Clearly $D(x)$ is the smallest when $\Phi^s(x)$ increases as much as possible
and $\Phi^{s'}(x)$ increases as little as possible,
where we consider a decrease as a negative increase.
This situation happens,
when the move $x\rightarrow x'$ causes the merger of ${\cal L}^s$-clusters,
which implies $d(v_j, v_k)\tau< \underline{w}_k/c$,
while it causes the merger of $v_k$ only to an existing ${\cal L}^{s'}$-cluster.\footnote{Note that
if $v_k$ doesn't merge to the cluster to its left under $s$ then it doesn't merge under $s'$ either.
}
Since $C^s_{L,j}(v_i)$ and $C^s_{L,j}(v_j)$ are two separate clusters,
we have 
\begin{equation}\label{eqn:nomerge}
d(v_i,v_j)\tau > \lambda(C^s_{L,j}(v_j))/c.
\end{equation} 
The part of $\Phi^s(x)$ that is affected by the move is
\begin{eqnarray}\label{eqn:costs1a}
\Phi^s(x)&:& \lambda(C^s_{L,j}(v_j)) d(v_j, x)\tau  + \frac{\lambda(C^s_{L,j}(v_j))(\lambda(C^s_{L,j}(v_j))+1)}{2c}\nonumber\\
		&+& \lambda(C^s_{L,i}(v_i)) d(v_i, x)\tau  + \frac{\lambda(C^s_{L,i}(v_i))(\lambda(C^s_{L,i}(v_i))+1)}{2c}. \nonumber
\end{eqnarray}
Since $C^s_{L,i}(v_i)$ and $C^s_{L,j}(v_j)$ are merged into $C^s_{L,k}(v_k)$ by assumption,
we have
\begin{equation}
d(v_i,v_k)\tau \leq \{\lambda(C^s_{L,j}(v_j)) + \underline{w}_k\}/c, 
\end{equation} 
and the part of $\Phi^{s'}(x)$ that is affected by the move is
\begin{eqnarray}\label{eqn:costs1b}
\Phi^s(x') &&: \lambda(C^s_{L,k}(v_k)) d(v_k, x')\tau \nonumber\\  
    &&+ \frac{\{\lambda(C^s_{L,i}(v_i)) + \lambda(C^s_{L,j}(v_j))+\underline{w}_k\}\{\lambda(C^s_{L,i}(v_i)+ \lambda(C^s_{L,j}(v_j))+\underline{w}_k+1\}}{2c},\nonumber
\end{eqnarray}
where
$\lambda(C^s_{L,k}(v_k)) =\lambda(C^s_{L,i}(v_i)) + \lambda(C^s_{L,j}(v_j))+\underline{w}_k$.
We now compute the increase
\begin{eqnarray}\label{eqn:costdiff1}
\Phi^s(x') - \Phi^s(x) &=& \lambda(C^s_{L,k}(v_k)) d(v_k, x')\tau - \lambda(C^s_{L,i}(v_i)) d(v_i, x)\tau \nonumber\\
 &-& \lambda(C^s_{L,j}(v_j)) d(v_j, x)\tau +  \lambda(C^s_{L,i}(v_i))\lambda(C^s_{L,j}(v_j))/c\nonumber\\
  &+& \{\lambda(C^s_{L,i}(v_i))+\lambda(C^s_{L,j}(v_j))\}\underline{w}_k/c
+ \underline{w}_k(\underline{w}_k + 1)/2c. 
\end{eqnarray}
Similarly, we have under $s'$,
\begin{eqnarray}\label{eqn:costs2}
\Phi^{s'}(x)&=& \lambda(C^{s'}_{L,j}(v_j)) d(v_j, x)\tau  + \frac{\lambda(C^{s'}_{L,j}(v_j))(\lambda(C^{s'}_{L,j}(v_j))+1)}{2c}\nonumber\\
\Phi^{s'}(x')&=& \lambda(C^{s'}_{L,k}(v_k)) d(v_k, x')\tau  
    + \frac{(\lambda(C^{s'}_{L,j}(v_j)) +\underline{w}_k)(\lambda(C^{s'}_{L,j}(v_j)) +\underline{w}_k+1)}{2c},\nonumber
\end{eqnarray}
where $ \lambda(C^{s'}_{L,k}(v_k)) =\lambda(C^{s'}_{L,j}(v_j)) +\underline{w}_k$,
and the increase is
\begin{eqnarray}\label{eqn:costdiff2}
\Phi^{s'}(x')&-&\Phi^{s'}(x)= \lambda(C^{s'}_{L,k}(v_k)) d(v_k, x')\tau - \lambda(C^{s'}_{L,j}(v_j)) d(v_j, x)\tau\nonumber\\
&+&  \lambda(C^{s'}_{L,j}(v_j))\underline{w}_k/c + \underline{w}_k(\underline{w}_k + 1)/2c.
\end{eqnarray}
We clearly have $\lambda(C^{s'}_{L,k}(v_k)) >\lambda(C^s_{L,k}(v_k))$,
and (\ref{eqn:nomerge}) implies $d(v_i, x)\tau>\lambda(C^s_{L,j}(v_j))$,
since $v_j \prec x$.
The assumption that $v_k$ is merged into $C^s_{L,x}(v_j)$ and $C^{s'}_{L,x}(v_j)$
implies $d(v_j, v_k)\tau < \underline{w}_k/c$ for $v_j\prec x\prec v_k$.
We conclude that
\[
\{\Phi^{s'}(x')-\Phi^{s'}(x)\} - \{\Phi^s(x')-\Phi^s(x)\} > 0,
\]
when $x\prec v_k \prec x'$.
This is valid in particular if $x=v^-_k$ and $x'=v^+_k$,
where $v^-$ (resp. $v^+$) denote a point on the left (resp. right) of $v$ that is arbitrarily close to $v$.
It is clear that this relation also holds, if $v_k$ is not merged into $C^s_{L,j}(v_j)$ and $C^{s'}_{L,j}(v_j)$.
\end{proof}
\end{document}